\documentclass[envcountsame]{llncs}
\usepackage[utf8]{inputenc}
\usepackage{amsfonts}
\usepackage{amsmath}
\usepackage{amssymb}
\usepackage{array}
\usepackage{caption}
\usepackage{clrscode3e}
\usepackage{dsfont}
\usepackage{enumerate}
\usepackage{fancyhdr}
\usepackage{geometry}
\usepackage{graphicx}
\usepackage{latexsym}
\usepackage{mdwlist}
\usepackage{pgfplots}
\usepackage{subfigure}
\usepackage{textcomp}
\usepackage{tikz}
\pgfplotsset{compat=newest}
\usetikzlibrary{patterns, intersections, calc}

\newcommand{\plspA}{\hspace{1.5ex}}
\newcommand{\plspB}{\hspace{2ex}}
\newcommand{\plspC}{\hspace{4ex}}
\newcommand{\eps}{\epsilon}
\newcommand{\Oh}{\mathrm{O}}
\newcommand{\APX}{\ensuremath{\mathrm{APX}}}

\newcommand{\sol}{\ensuremath{\mathrm{sol}}}

\newcommand{\STM}{\ensuremath{\mathrm{STM}}}
\newcommand{\Improved}{\ensuremath{\mathrm{I}}}
\newcommand{\Loose}{\ensuremath{\mathrm{L}}}
\newcommand{\Profit}{\ensuremath{\mathrm{P}}}
\newcommand{\FFFS}{\ensuremath{\mathrm{U}}}
\newcommand{\Utility}{\ensuremath{\mathrm{U}}}
\newcommand{\relaxation}{\ensuremath{\mathit{LR}}}
\newcommand{\Ball}{\ensuremath{\mathit{Ball}}}

\newcounter{constraint}
\newcommand{\constraintlabel}[1]{\refstepcounter{constraint}\label{#1}(\theconstraint)}

\newenvironment{linearprogram}
                {\noindent\begin{tabular}{l@{\plspC}l@{\plspB}%
                              >{\begin{math}\displaystyle}r<{\end{math}}@{\plspA}%
                              >{\begin{math}\displaystyle}c<{\end{math}}@{\plspA}%
                              >{\begin{math}\displaystyle}l<{\end{math}}@{\plspB}%
                              >{\begin{math}\displaystyle}l<{\end{math}}@{\plspB}%
                              l@{}}}
                {\end{tabular}\noindent}%
\tikzstyle{vertex}=[draw, circle,minimum size=10pt,inner sep=1pt]
\tikzstyle{edge} = [draw,line width=0.5pt]
\tikzstyle{selected edge} = [draw,line width=1.5pt]
\tikzstyle{weight} = [font=\scriptsize, fill=white, inner sep=0, text=black, circle]
\newcounter{aux}

\begin{document}
\mainmatter
\title{The Unit-Demand Envy-Free Pricing Problem$^\star$}
\author{Cristina G. Fernandes%
\and Carlos E. Ferreira%
\and \\ Álvaro J. P. Franco%
\and Rafael C. S. Schouery%
}%
\institute{Department of Computer Science, University of São Paulo, Brazil\\
\email{\{cris,cef,alvaro,schouery\}@ime.usp.br}
}

\footnotetext[1]{Research partially supported by CAPES (Proc.~\mbox{33002010176P0}), 
  CNPq (Proc.~302736/2010-7, 308523/2012-1, and 477203/2012-4), 
  FAPESP (Proc.~2009/00387-7), and Project MaCLinC of NUMEC/USP.}
 
\maketitle

\begin{abstract}
  We consider the unit-demand envy-free pricing problem, which is a unit-demand auction where each bidder receives an item that maximizes his utility, and the goal is to maximize the auctioneer's profit. This problem is NP-hard and unlikely to be in APX. We present four new MIP formulations for it and experimentally compare them to a previous one due to Shioda, Tun\c cel, and Myklebust. We describe three models to generate different random instances for general unit-demand auctions, that we designed for the computational experiments. Each model has a nice economic interpretation. Aiming approximation results, we consider the variant of the problem where the item prices are restricted to be chosen from a geometric series, and prove that an optimal solution for this variant has value that is a fraction (depending on the series used) of the optimal value of the original problem. So this variant is also unlikely to be in APX.
\end{abstract}

\keywords{pricing problem, envy-free allocations, unit-demand auctions} 

\section{Introduction}
A part of mechanism design is devoted to profit maximization. The fundamental problem of maximizing profit while selling multiple items to consumers has been vastly considered in the literature and can be viewed as a type of auction. Hartline and Karlin~\cite[Chap.\ 13]{NisanRTV07} provided an introduction on this class of problems. The general goal is to determine a price for each item and an allocation of items to consumers in a way that the overall profit of the seller is maximized. However, once the prices are defined, the consumers may behave in several ways. In the literature, there are a variety of models (max/min-buying, rank-buying, max-gain-buying)~\cite{AFMZ04} that try to capture the usual behavior of consumers.

Guruswami et al.~\cite{GHKKKM05} formalized one such model, suggested by Aggarwal et al.~\cite{AFMZ04}, defining the \emph{envy-free pricing} for combinatorial auctions. In the more general setup, each bidder (consumer) is willing to buy a bundle of items. Knowing the valuation of every bundle of items for each bidder, the auctioneer (seller) has to decide on the pricing and on an allocation of the items to the bidders in an \emph{envy-free manner}, that is, in a way that any bidder is at least as happy with the bundle of items assigned to him (that might even be empty) as with any other bundle of items, considering the prices. The auctioneer's goal is to maximize his profit. Of course, in this general form, there is an issue on the amount of information involved, specifically for the valuations. So it is reasonable to consider, as Guruswami et al., particular cases that avoid this issue.

A well-studied case is the \emph{unit-demand} envy-free pricing problem, in which each bidder is willing (or is allowed) to buy at most one item, and there is an unlimited supply of each item. Guruswami et al.\ presented a $(2\ln n)$-approximation for this problem, where $n$ is the number of bidders, and proved that this problem is \APX-hard (with each bidder valuation being $1$ or $2$ for the items of interest). Briest~\cite{Briest08} further analyzed the hardness of approximation, considering the uniform budget case, where each bidder equally values all items he is interested in. Assuming specific hardness of the balanced bipartite independent set problem in constant degree graphs, or hardness of refuting random 3CNF formulas, there is no approximation for the uniform budget unit-demand envy-free pricing problem with ratio $\Oh(\lg^{\epsilon} n)$ for some $\epsilon > 0$. Chen and Deng~\cite{CD10} showed that the unit-demand envy-free pricing problem can be solved in polynomial time if every bidder has positive valuation for at most two items.

More recently, Shioda, Tun\c cel, and Myklebust~\cite{ShiodaTM11} described a slightly more general model than the one of Guruswami et al.~\cite{GHKKKM05}, and presented a mixed-integer programming (MIP) formulation for their model, along with heuristics and valid cuts for their formulation.

\subsection{Our Results} 

We present new MIP formulations for the unit-demand envy-free problem that can be adapted also to the model described by Shioda et al.~\cite{ShiodaTM11}. We compare these new formulations to the one of Shioda et al.\ through computational experiments that indicate that our new formulations give better results in practice than the other one.

For the computational experiments, we used six different sets of instances, obtained from generators we designed. Leyton-Brown et al.~\cite{LeytonBrown00} presented the well-known Combinatorial Auction Test Suite (CATS) generator, proposed to provide realistic bidding behavior for five real world situations. See also~\cite[Chap.\ 18]{CramtonSS06}. Their generator produces instances for auctions where the bidders are interested in buying a bundle of items. A simple adaptation of CATS for unit-demand auctions does not preserve the economic motivations. So we designed three models to generate different random instances for unit-demand auctions, and we implemented three corresponding generators. Each of the three models has a nice economic interpretation and can be used in other works on unit-demand auctions, such as~\cite{MyklebustST12,ShiodaTM11}. We believe these instances generators are a nice contribution by themselves, and we made them available as open-source software at \texttt{https://github.com/schouery/unit-demand-market-models}.

The inapproximability result of Briest~\cite{Briest08} indicates that the unit-demand envy-free pricing problem is unlikely to be in APX. In the search for an APX version of the problem, we considered a variant in which the prices are restricted to assume only specific values (that forms a geometric series with ratio $1+\epsilon$ for a fixed $\epsilon > 0$). We prove that any constant approximation for this variant is a constant (depending on~$\eps$) approximation for the original problem. This implies that this variant is also unlikely to be in APX. Yet, the connection between the two problems seems interesting. At first, one can think that the relation is not surprising but, by changing even a bit the prices of the items, an envy-free allocation may change dramatically, incurring in a big loss of profit. This result however assures that the loss (in terms of optimal values) is within a constant factor.

The paper is organized as follows. In the next section, we present some notation and describe formally the problem that we address.  In Sect.~\ref{sec:mip}, we present our new MIP formulations for the problem and revise the one by Shioda et al.~\cite{ShiodaTM11}. In Sect.~\ref{sec:models}, we present the three instances generators and describe the economical motivation behind each one. In Sect.~\ref{sec:empirical_results}, we empirically compare our formulations for the unit-demand envy-free pricing problem against the one by Shioda et al.~\cite{ShiodaTM11}, using sets of instances produced by our generators. In Sect.~\ref{sec:particular_case}, we consider the variant of the problem with restricted prices and show that this variant is still hard to approximate, that is, if one would find a constant factor approximation for this variant, then this approximation would also have a constant factor for the unit-demand envy-free pricing problem. We conclude with some final remarks in Sect.~\ref{sec:conclusion}.

\section{Model and Notation}\label{sec:models-and-notation}

We denote by $B$ the set of bidders and by $I$ the set of items.  

\begin{definition}
  A \emph{valuation} matrix is a non-negative rational matrix $v$ indexed by $I \times B$. 
\end{definition}

The number $v_{ib}$ represents the value of item $i$ to bidder $b$. 

\begin{definition}
  A \emph{pricing} is a non-negative rational vector indexed by $I$.
\end{definition}

\begin{definition}
  A \emph{(unit-demand) allocation} is a binary matrix $x$ indexed by~${I \times B}$,
  such that, for every bidder $b$, $x_{ib} = 1$ for at most one item $i$.
\end{definition}

If $x_{ib} = 1$, we say that $i$ is sold to $b$ and $b$ receives $i$.  If $x_{ib} = 0$ for every item~$i$, we say that~$b$ receives no item.  If $x_{ib} = 0$ for every bidder $b$, we say that $i$ is not sold. Note that an allocation is not necessarily a matching, as the same item can be assigned to more than one bidder (each one receives a copy of the item).

Fix a valuation matrix $v$ and a pricing $p$. The next definitions formalize the concept of an envy-free allocation.

\begin{definition}
  For an allocation $x$, the \emph{utility} of a bidder $b$, denoted by $u_b$, is $v_{ib} - p_i$ if bidder $b$ receives item $i$ and $0$ if bidder $b$ receives no item.
\end{definition}

\begin{definition}
  For an allocation $x$, bidder $b$ is \emph{envy-free} if $u_b \geq 0$ and $u_b \geq v_{ib} - p_i$ for every item~$i$. We say that $x$ is \emph{envy-free} if every $b$ in $B$ is envy-free.
\end{definition}

In other words, there is no item that would give a bidder a better utility. Now we are ready to state the problem.

\begin{definition}
  The \emph{unit-demand envy-free pricing problem} consists of, given a valuation matrix $v$, finding a pricing $p$ and an envy-free allocation $x$ that maximize the auctioneer's profit, which is the sum of the prices of items received by the bidders (considering multiplicities).
\end{definition}

Fig.~\ref{fig:envy-free-example}(a) illustrates an instance of the problem, and Fig.~\ref{fig:envy-free-example}(b) shows an envy-free allocation $x$ for the pricing $p = (5, 8, 3)$, where $x_{12} = x_{13} = x_{31} = 1$, while all other entries of $x$ are $0$. The profit of the auctioneer in this case is 13. Note that this is not the best possible profit, as there is an envy-free allocation for the pricing $p = (6,6,3)$ that leads to profit 21 (take $x_{13} = x_{22} = x_{24} = x_{31} = 1$).
\begin{figure}[tb]
\centering
\begin{tikzpicture}[x = 40pt, y=-30pt, scale=0.8, every node/.style={scale=0.80}]
  \def\largura{2}
  \def\xlabel{-40pt}
  \def\ylabel{-15pt}
  \begin{scope}
  \node at (\xlabel,\ylabel) {(a)};
  \foreach \i in {1,2,3} {
    \node[vertex] (i\i) at (0,\i) {$i_\i$};
  }
  \foreach \b in {1,2,3,4} {
    \pgfmathsetmacro{\y}{\b-0.5}
    \node[vertex] (b\b) at (\largura,\y) {$b_\b$};
  }
  \foreach \s/\d/\w in {
    1/1/4,
    1/2/5,
    1/3/6,
    2/2/7,
    2/4/6,
    3/1/3,
    3/3/2,
    3/4/2} {
      \def\pos{0.37}
      \draw[edge, name path={i\s--b\d}] (i\s) -- (b\d);
      \path[name path=vertical] (\pos,0) -- (\pos,4);
      \path[name intersections={of={i\s--b\d} and vertical}];
      \node[weight] at (intersection-1) {$\w$};
  }
  \end{scope}

  \begin{scope}[xshift = 200]
  \node at (\xlabel,\ylabel) {(b)};
  \foreach \i/\w in {1/5,2/8,3/3} {
    \node[vertex, label={left:$\w$}] (i\i) at (0,\i) {$i_\i$};
  }

  \foreach \b in {1,2,3,4} {
    \pgfmathsetmacro{\y}{\b-0.5}
    \node[vertex] (b\b) at (\largura,\y) {$b_\b$};
  }
  \foreach \s/\d/\w/\style in {
    1/1/4/edge,
    1/2/5/selected edge,
    1/3/6/selected edge,
    2/2/7/edge,
    2/4/6/edge,
    3/1/3/selected edge,
    3/3/2/edge,
    3/4/2/edge} {
      \def\pos{0.37}
      \draw[\style, name path={i\s--b\d}] (i\s) -- (b\d);
      \path[name path=vertical] (\pos,0) -- (\pos,4);
      \path[name intersections={of={i\s--b\d} and vertical}];
      \node[weight] at (intersection-1) {$\w$};
  }
  \end{scope}
\end{tikzpicture}
\caption{(a) A valuation matrix represented by a weighted bipartite graph whose edge weights give the non-null valuations. (b) An envy-free allocation, given by the darker edges, for the pricing shown at the left of the items.}\label{fig:envy-free-example} 
\end{figure}
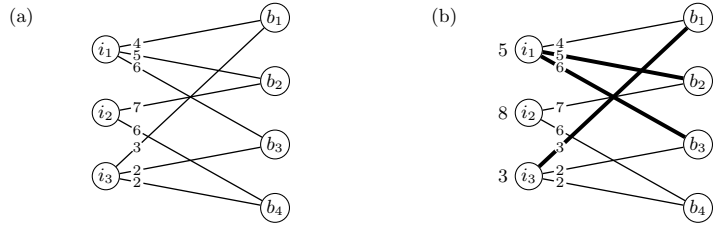

Given a valuation matrix $v$ and a pricing $p$, it is easy to compute an envy-free allocation $x_p$ that maximizes the auctioneer's profit. Indeed, for each bidder~$b$, assign to $b$ an item $i$ that maximizes $v_{ib} - p_{i}$ (if that is non-negative), resolving ties by choosing the item with higher price (if ties persist, there is more than one such allocation). Hence, to solve the problem, it is enough to find a pricing $p$ such that $p$ and $x_p$ maximize the auctioneer's profit.

\section{MIP Formulations}\label{sec:mip}

Before presenting the formulations, let us denote, for an item $i$, $\max\{v_{ib} : b \in B\}$ as $R_i$ and, for a bidder $b$, $\max\{v_{ib} : i \in I\}$ as $S_b$. These definitions will be used for ``big-M'' inequalities in the formulations presented in this section. 
We start by presenting the formulation due to Shioda et al.~\cite{ShiodaTM11} slightly adapted to our problem. Consider a valuation matrix $v$. 

We use the following variables: a binary matrix $x$ indexed by $I \times B$ that represents an allocation, a rational vector $p$ indexed by $I$ that represents the pricing  and a rational matrix $\hat{p}$ indexed by $I \times B$ that represents the price paid for item $i$ by each bidder $b$ (that is, $\hat{p}_{ib} = 0$ if $b$ does not receive~$i$ and $\hat{p}_{ib} = p_i$ otherwise). The formulation, which we name (\STM), consists of, finding $x$, $p$, and $\hat{p}$ that
\begin{center}
\begin{linearprogram}
  (\STM) &\mbox{maximize} & \multicolumn{3}{l}{$\displaystyle\sum_{b \in B} \sum_{i \in I} \hat{p}_{ib}$}\\[1.5ex]
  &\mbox{subject to} & \sum_{i \in I} x_{ib} & \leq & 1, & \forall b \in B\\
  && \sum_{i \in I \setminus \{k\}} (v_{ib}x_{ib} - \hat{p}_{ib}) & \geq & v_{kb}\!\!\sum_{i \in I \setminus \{k\}} \!\!\!\!x_{ib} \ - \ p_k, & \forall k \in I, \forall b \in B  & \constraintlabel{s:envyfree}\\
  && v_{ib}x_{ib} - \hat{p}_{ib} & \geq & 0, & \forall b \in B, \forall i \in I\\
  && \hat{p}_{ib} & \leq & p_i, & \forall b \in B, \forall i \in I\\
  && \hat{p}_{ib} & \geq & p_i - R_i(1 - x_{ib}), & \forall b \in B, \forall i \in I\\
  &&\multicolumn{3}{c}{$\displaystyle x_{ib} \in \{0,1\},\quad p_i \geq 0,\quad \hat{p}_{ib} \geq 0 $} & \forall b \in B, \forall i \in I.
\end{linearprogram}
\end{center}

From the formulation (\STM) we developed a new and stronger formulation by changing inequalities (\ref{s:envyfree}). We call this improved formulation (\Improved), which we present below.
\begin{center}
\begin{linearprogram}
  (\Improved) & \mbox{maximize} & \multicolumn{3}{l}{$\displaystyle\sum_{b \in B} \sum_{i \in I} \hat{p}_{ib}$}\\[1.5ex]
  &\mbox{subject to} & \sum_{i \in I} x_{ib} & \leq & 1, & \forall b \in B\\
  && \sum_{i \in I} (v_{ib}x_{ib} - \hat{p}_{ib}) & \geq & v_{kb} - p_k, & \forall k \in I, \forall b \in B\\
  && v_{ib}x_{ib} - \hat{p}_{ib} & \geq & 0, & \forall b \in B, \forall i \in I\\
  && \hat{p}_{ib} & \leq & p_i, & \forall b \in B, \forall i \in I & \constraintlabel{i:redundant} \\
  && \hat{p}_{ib} & \geq & p_i - R_i(1 - x_{ib}), & \forall b \in B, \forall i \in I\\
  &&\multicolumn{3}{l}{$\displaystyle x_{ib} \in \{0,1\},\quad p_i \geq 0,\quad \hat{p}_{ib} \geq 0 $} & \forall b \in B, \forall i \in I.
\end{linearprogram}
\end{center}

Also, we noticed that inequalities (\ref{i:redundant}) are unnecessary, which led us to another formulation that we call (\Loose) (from the word \emph{loose}), consisting of the omission of inequalities (\ref{i:redundant}) from formulation (\Improved). That is, it is possible that (\Loose) is a weaker version of (\Improved) in terms of relaxation guarantees.

Studying formulation (\Loose), we developed another formulation with variables $x$ and $p$ as defined before and a rational vector $z$ indexed by $B$ that represents the profit obtained from bidder $b$ (that is, if bidder $b$ receives item $i$ then $z_b = p_i$ and $z_b = 0$ if $b$ does not receive an item). That is, one can view $z$ as $z_b = \sum_{i \in I} \hat{p}_{ib}$. This formulation, which we name (\Profit) (from the word \emph{profit}), consists of finding $x$, $p$, and $z$ that
\begin{center}
\begin{linearprogram}
  (\Profit) & \mbox{maximize} & \multicolumn{3}{l}{$\displaystyle\sum_{b \in B} z_b$}\\[1.5ex]
  &\mbox{subject to} & \sum_{i \in I} x_{ib} & \leq & 1, & \forall b \in B\\
  && \sum_{i \in I} v_{ib}x_{ib} - z_b & \geq & v_{kb} - p_k, & \forall k \in I, \forall b \in B\\
  && \sum_{i \in I} v_{ib}x_{ib} - z_b & \geq & 0, & \forall b \in B\\
  && z_b & \geq & p_i - R_i(1 - x_{ib}), & \forall b \in B, \forall i \in I\\
  &&\multicolumn{3}{l}{$\displaystyle x_{ib} \in \{0,1\},\quad p_i \geq 0,\quad z_b \geq 0 $} & \forall b \in B, \forall i \in I.
\end{linearprogram}
\end{center}

As one would expect, (\Profit) is a weaker formulation than (\Loose).

Finally, we present our last formulation that was developed changing the focus from the price paid by a bidder to the utility that a bidder has. We use the following variables: $x$ and $p$ as defined before and a rational vector $u$ indexed by $B$ that represents the utilities of the bidders. The formulation, which we name (\FFFS) (from the word \emph{utility}), consists of, finding $x$, $p$, and $u$ that
\begin{center}
\begin{linearprogram}
	(\Utility) & \mbox{maximize} & \multicolumn{3}{l}{$\displaystyle\sum_{i \in I}\sum_{b \in B} v_{ib}x_{ib} - \sum_{b \in B} u_b$} \\[1.5ex]
	&\mbox{subject to} & \sum_{i \in I} x_{ib} &\leq& 1, & \forall b \in B\\
	&&u_b &\geq& v_{ib} - p_i, & \forall i \in I, \ \forall b \in B\\
	&&u_b &\leq& v_{ib}x_{ib} - p_i + (1 - x_{ib})(R_i + S_b), & \forall i \in I, \ \forall b \in B\\
	&&u_b &\leq& \sum_{i \in I}v_{ib}x_{ib}, & \forall b \in B\\
&&\multicolumn{3}{l}{$\displaystyle x_{ib} \in \{0,1\},\quad p_i \geq 0,\quad u_b \geq 0 $} & \forall b \in B, \forall i \in I.  
\end{linearprogram}
\end{center}

For some MIP formulation (F) and a valuation $v$, we denote by $\relaxation_{F}(v)$ the value of an optimal solution to the linear relaxation of (F) when the instance is $v$. Our next result formalizes the relaxation quality guarantees for the formulations presented in this section. 

\newcounter{polyhedra}
\setcounter{polyhedra}{\value{theorem}}
\begin{theorem}\label{thm:polyhedra}
  For every instance $v$ of the unit-demand envy-free problem, we have that $\relaxation_{\Improved}(v)~\leq~\relaxation_{\STM}(v)$ and $\relaxation_{\Improved}(v) \leq \relaxation_{\Loose}(v) \leq \relaxation_{\Profit}(v) \leq \relaxation_{\Utility}(v)$. Also, there is an instance where $\relaxation_{\Improved}(v) < \relaxation_{\STM}(v)$ and $\relaxation_{\Improved}(v) = \relaxation_{\Loose}(v) < \relaxation_{\Profit}(v) < \relaxation_{\Utility}(v)$.
\end{theorem}

We do not know if there is an instance $v$ where $\relaxation_{\Improved}(v) < \relaxation_{\Loose}(v)$.

Before presenting and discussing the experimental comparison of the these formulations, we describe the
instances generators that we used in the computational experiments. 

\section{Multi-Item Auctions Test Suites}\label{sec:models}

In this section, we propose three new models to generate different random instances for unit-demand
auctions, each of them with a nice economic interpretation. 

In these three models, we start by creating a bipartite graph where one side of the bipartition
represents the items and the other, the bidders. An edge of this graph represents that the bidder 
gives a positive value to the corresponding item. 


\subsection{Characteristics Model}

The idea is to create a graph where each item has a profile of characteristics and each bidder
desires some of these characteristics. In this way, two items that have the same profile are desired
by the same bidders, and bidders desire items that do not differ much.

It is natural that an item has a market value, and the valuations for that item (by interested bidders) 
are concentrated around this value. That is, the valuations for an item should not differ too much, and 
valuations that are far away from this market value are rare.

An instance of our problem, that is, a valuation matrix, will be represented by a weighted
bipartite graph as we describe in the following. 

Let $m$ be the number of items, $n$ be the number of bidders, $c$ be the number of the
characteristics of an item, $o$ be the number of options for any characteristic, $p$ be the
number of options preferred by a bidder for any characteristic, $\ell$ be the minimum market
value of an item, $h$ be the maximum market value of an item, and $d$ be the percentage of
deviation used. The set of vertices is $I \cup B$ and the set of edges $E$ is described below.

For every item $i$, we have a vector of size $c$ (the characteristics of that item) where every entry is in $\{1,\ldots, o\}$ chosen independent and uniformly at random. For every bidder $b$, we construct a matrix $A \in \{0,1\}^{c \times p}$ where, for every row, we choose independent and uniformly at random $p$ positions to be set to~$1$ and $o-p$ positions to be set to~$0$. For a bidder $b$ and an item~$i$, we have that $ib \in E$ if and only if the
characteristics of item~$i$ coincide with the preferences of bidder $b$, that is, if the characteristic $k$ of item $i$ has value $v$ then $A_{kv} = 1$ where $A$ is the matrix of bidder $b$. See an example
of such characteristics and preferences in Fig.~\ref{fig:charac-model}(a).

Finally, for every item $i$, we define $\bar{p}_i$ as the market price of item $i$, chosen independent and 
uniformly at random from the interval $[\ell,h]$. For every $ib$ in $E$, we choose $v_{ib}$ from $1.0 +
\mathcal{N}(\bar{p}_i, (\bar{p}_i d)^2)$, where $\mathcal{N}(\mu, \sigma^2)$ denotes the Gaussian
distribution with mean $\mu$ and standard deviation $\sigma$. See Fig.~\ref{fig:charac-model}(b).
\begin{figure}[htb] 
\centering
\begin{tikzpicture}[x = 40pt, y=-30pt]
  \tikzstyle{weight} = [font=\scriptsize, fill=white, inner sep=1, text=black]
  \def\largura{2}
  \def\xlabel{-45pt}
  \def\ylabel{-25pt}
  \def\altura{1}
  \begin{scope}
      \node at (\xlabel,\ylabel) {(a)};
      \foreach \i/\l in {
      1/{$\left[\begin{array}{c} 3 \\ 1 \end{array}\right]$}, 
      2/{$\left[\begin{array}{c} 3 \\ 2 \end{array}\right]$},
      3/{$\left[\begin{array}{c} 2 \\ 1 \end{array}\right]$}}{
        \node[vertex, label={left:\l}] (i\i) at (0,\i) {$i_\i$};
      }
      \foreach \b/\h/\l in {
      1/{1&0&1}/{1&0&1},
      2/{1&0&1}/{1&1&0},
      3/{0&1&1}/{1&0&1}}
      {
        \node[vertex, label={right:$\left[\begin{array}{ccc}\h\\ \l \end{array}\right]$}] (b\b) at (\largura,\b) {$b_\b$};
      }
      \foreach \source/ \dest in {
      i1/b1,
      i1/b2,
      i1/b3,
      i2/b2,
      i3/b3}
      {
        \draw[edge] (\source) -- (\dest);
      }
  \end{scope}
  \begin{scope}[xshift = 200]
    \node at (\xlabel,\ylabel) {(b)};
    \foreach \i/\price in {1/183,2/150,3/170}
    {
      \node[vertex, label={left:$\price$}] (i\i) at (0,\i) {$i_\i$};
    }
    \foreach \b in {1,2,3}
    {
      \node[vertex] (b\b) at (\largura,\b) {$b_\b$};
    }
    \foreach \source/ \dest / \weight / \y in {
    i1/b1/185,
    i1/b2/177,
    i1/b3/182,
    i2/b2/160,
    i3/b3/167}
    {
      \def\pos{0.5}
      \draw[edge, name path={\source--\dest}] (\source) -- (\dest);
      \path[name path=vertical] (\pos,1) -- (\pos,3);
      \path[name intersections={of={\source--\dest} and vertical}];
      \node[weight] at (intersection-1) {$\weight$};
    }
  \end{scope}
\end{tikzpicture}
\caption{Instance with two characteristics, each with three options. Each bidder wants items that have one of two options (given by their matrices) in every characteristic. The market price of each item was chosen from the interval $[140, 190]$, and the percentage of deviation used was $d = 10$.} 
\label{fig:charac-model} 
\end{figure}
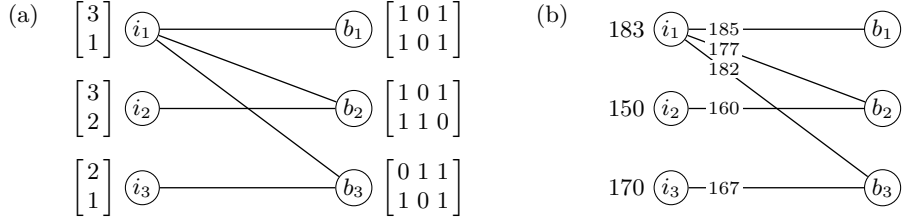

\subsection{Neighborhood Model} 

This model selects the valuations based on a geometric relation.  For this, we adapt a random
process first introduced by Gilbert~\cite{Gilbert61} for generating a Random Geometric Graph.

The idea is to distribute the items and the bidders in a $1 \times 1$ square in $\mathds{Q}^2$ and
define the valuations according to the distance between a bidder and an item.

We denote the $1 \times 1$ square in $\mathds{Q}^2$ by $W$. A ball (in $\mathds{Q}^2$) with center
at position $(x, y)$ and radius~$r$ is denoted by $\Ball(x, y, r)$.

We are given the number of items $m$, the number of bidders $n$, and a radius~$r$. In order to
construct the valuation matrix, we construct a graph $G=(I \cup B, E)$ such that $|I| = m$ and $|B|
= n$. First we assign a point in $W$ independent and uniformly at random to each item and each bidder. 
That is, item $i$ and bidder $b$ are represented in $W$ by the points $(x_i, y_i)$ and $(x_b, y_b)$ in
$\mathds{Q}^2$, respectively. There is an edge $ib$ in $E$ if $(x_i, y_i) \in \Ball(x_b, y_b, r)$, that
is, if the distance between $i$ and $b$ is at most~$r$.

Finally, to define the valuations, for every bidder $b$, we choose a multiplier $k_b$ in $[1,h]$
where~$h$ is given as input. If $ib \in E$, then $v_{ib} = 1.0 + Mk_b/d(i,b)$ where $d(i,b)$
denotes the (Euclidean) distance between $i$ and $b$, and $M$ is a scaling factor also given as input.

This way, a bidder only gives a positive valuation to items sufficiently close to his location,
and the valuation decreases as the distance increases. Two bidders that are at the same distance
from an item still can evaluate it differently because they might have different $k$ multipliers.

\subsection{Popularity Model}

In this model we capture the idea that a market might have popular and unpopular items.

Let $m$ be the number of items, $n$ be the number of bidders, $e$ be the number of edges, $Q$ be the
maximum quality of an item, and $d$ be a percentage of deviation. We construct the bipartite graph
as follows. We start with the empty bipartite graph with parts $I$ and $B$ such that $|I| = m$ and
$|B| = n$. We add edges at random (as described below) until the graph has $e$ edges.

To add an edge, we choose a bidder uniformly at random, and choose an item biased according to its
degree. That is, an item~$i$ of current degree $d_i$ has a weight of ${d_i + 1}$.
(This process, called Preferential Attachment, was already used in the literature by Barabási and
Albert~\cite{BarabasiA99} to generate graphs with properties that appear in the World Wide Web and
in graphs that arise in biology, such as the interactions of the brain cells.)

This way, items that are ``popular'', when we are adding an edge, have a higher chance to be chosen
than ``unpopular'' items. The final result is a graph with a small number of items of large degree
and a large number of items of small degree.

We also use the information of the (final) degree of an item to decide on the valuations. As in the
Characteristic Model, an item $i$ has a market price $\bar{p}_i$, and we choose the valuation
$v_{i,b}$ according to $1.0 + \mathcal{N}(\bar{p}_i, (\bar{p}_i d)^2)$. But, in contrast to the
Characteristic Model, we do not choose the market value uniformly at random in an interval.  
For each item $i$, we choose a random quality $q_i$ in $(0, Q]$ and we define $\bar{p}_i$ as
$q_i/d_i$. This way, the market price of an item increases with its quality and decreases with its
degree. That is, an item has high desirability if it is ``cheap'' or has good quality. 
See Fig.~\ref{fig:pop-model}.
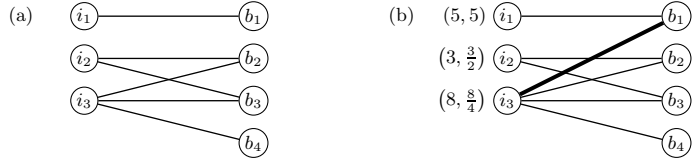
\begin{figure}[tb] 
\begin{center}
\begin{tikzpicture}[x = 40pt, y=-20pt, scale=0.8, every node/.style={scale=0.8}]
  \def\xlabel{-50pt}
  \def\ylabel{-20pt} 
  \def\largura{2}
  \begin{scope}
    \node at (-30pt,\ylabel) {(a)};  
    \foreach \i in {1,2,3}
    {
      \node[vertex] (i\i) at (0,\i) {$i_\i$};
    }
    \foreach \b in {1,...,4}
    {
      \pgfmathsetmacro{\h}{\b}
      \node[vertex] (b\b) at (\largura,\h) {$b_\b$};
    }
    \foreach \source/ \dest in {
      i1/b1,
      i2/b2,
      i2/b3,
      i3/b2,
      i3/b3,
      i3/b4}
    {
      \draw[edge] (\source) -- (\dest);
    }
  \end{scope}
  \begin{scope}[xshift=200]
    \node at (\xlabel,\ylabel) {(b)};  
    \foreach \i/\q/\p in {1/5/5,2/3/\frac{3}{2},3/8/\frac{8}{4}}
    {
      \node[vertex, label={left:$\left(\q, \p\right)$}] (i\i) at (0,\i) {$i_\i$};
    }
    \foreach \b in {1,...,4}
    {
      \pgfmathsetmacro{\h}{\b}
      \node[vertex] (b\b) at (\largura,\h) {$b_\b$};
    }
    \foreach \source/ \dest in {
      i1/b1,
      i2/b2,
      i2/b3,
      i3/b2,
      i3/b3,
      i3/b4}
    {
      \draw[edge] (\source) -- (\dest);
    }
    \draw[selected edge] (i3)-- (b1);
  \end{scope}
\end{tikzpicture}

\caption{From (a) to (b), the $7^{th}$ edge $i_3b_1$ was added. The bidder $b_1$ was chosen uniformly and item~$i_3$ was chosen because of its high degree. The pair within parenthesis in (b) represents $(q_i, \bar{p}_i)$. Each~$q_i$ was randomly chosen in the interval $(0,10]$.} 
\label{fig:pop-model} 
\end{center}
\end{figure}

\section{Empirical Results}\label{sec:empirical_results}

In this section, we present some empirical results involving the formulations. In our experiments, the new formulations produced better results than (\STM), the previously known formulation from the literature. In particular, formulation (\Loose) was the best for small instances, and formulation (\Utility) was the best for large instances (providing a small gap).


We used a computer with two Intel Xeon E5620 2.40 GHz processors, 64GB of RAM running Gentoo Linux
64bit. The MIPs were solved by CPLEX~12.1.0 in single-thread mode, with work memory limited
to~4GB. The tests were run using a time limit of one hour of CPU time.

For our tests, we generated six test sets, two for every one of our three models: one with 20 instances with 50, 100, 150, 200, 250, and~300 bidders and items (that we call ``small instances'') and other with 20 instances with 500, 1000, 1500, 2000, 2500 and~3000 bidders and items (that we call ``large instances''). Next we describe the parameters given to the models, and we denote by $n$ the number of items (that is equal to the number of bidders).

For the Characteristic Model, we chose the number of options for a characteristic as $o = 8$, the number of options prefered by a bidder for any characteristic as  $p = 7$ and the number of characteristics of an item as $\lceil\log(8/n) / \log(p/o)\rceil$, the minimum market value of an item as $\ell = 1$, the maximum market value of an item as $h = 100$, and the percentage of deviation as $d = 0.25$. Thus, the mean degree of the bidders is between $7$ and $8$. For the Neighborhood Model, we chose the maximum multiplier of a bidder as $h = 3$, the radius as $\sqrt{\frac{8}{n\pi}}$ (in this way, the mean degree of the bidders is near~$8$) and the scaling factor $M = 10$. Finally, for the Popularity Model, we chose the number of edges as $e = 8n$, the maximum quality of an item as $Q = 200$, and the percentage of deviation as $d = 0.25$.

Note that all of our instances are sparse, since we believe that in practice dense instances should
be rare. That is, we expect that, in a market with many items, each bidder would have positive
valuation usually only for a few items.

\begin{figure}[htb!]\centering
  \subfigure[Characteristics]{
    \begin{tikzpicture}
      \begin{axis}[ybar,
                   ymin = 0,
                   ymax = 22,
                   width = 8cm,
                   height = 4cm,
                   bar width = 0.15cm,
                   xtick = data,
                   cycle list = {black, black!70, black!30, black, black}
                   ]
        \addplot+[fill, pattern = crosshatch, restrict x to domain=50:250] table[x=n,y=optimal] {si-table_characteristics-S};
        \addplot+[fill] table[x=n,y=optimal, restrict x to domain=50:250] {si-table_characteristics-I};
        \addplot+[fill] table[x=n,y=optimal, restrict x to domain=50:250] {si-table_characteristics-L};
        \addplot+[fill, pattern = grid, restrict x to domain=50:250] table[x=n,y=optimal] {si-table_characteristics-P};
        \addplot+[fill] table[x=n,y=optimal, restrict x to domain=50:250] {si-table_characteristics-F};
      \end{axis}
  \end{tikzpicture}
  }
  \subfigure[Neighborhood]{
    \begin{tikzpicture}
      \begin{axis}[ybar,
                   ymin = 0,
                   ymax = 22,
                   width = 8cm,
                   height = 4cm,
                   bar width = 0.15cm,
                   xtick = data,
                   cycle list = {black, black!70, black!30, black, black}
                   ]
        \addplot+[fill, pattern = crosshatch, restrict x to domain=50:250] table[x=n,y=optimal] {si-table_neighborhood-S};
        \addplot+[fill] table[x=n,y=optimal, restrict x to domain=50:250] {si-table_neighborhood-I};
        \addplot+[fill] table[x=n,y=optimal, restrict x to domain=50:250] {si-table_neighborhood-L};
        \addplot+[fill, pattern = grid, restrict x to domain=50:250] table[x=n,y=optimal] {si-table_neighborhood-P};
        \addplot+[fill] table[x=n,y=optimal, restrict x to domain=50:250] {si-table_neighborhood-F};
      \end{axis}
  \end{tikzpicture}
  }
  \subfigure[Popularity]{
    \begin{tikzpicture}
      \begin{axis}[ybar,
                   ymin = 0,
                   ymax = 22,
                   width = 8cm,
                   height = 4cm,
                   bar width = 0.15cm,
                   xtick = data,
                   cycle list = {black, black!70, black!30, black, black},
                   legend style={legend pos=outer north east},
                   every axis legend/.append style={nodes={right}}
                   ]
        \addplot+[fill, pattern = crosshatch, restrict x to domain=50:250] table[x=n,y=optimal] {si-table_popularity-S};
        \addlegendentry{(\STM)}
        \addplot+[fill] table[x=n,y=optimal, restrict x to domain=50:250] {si-table_popularity-I};
        \addlegendentry{(\Improved)}
        \addplot+[fill] table[x=n,y=optimal, restrict x to domain=50:250] {si-table_popularity-L};
        \addlegendentry{(\Loose)}
        \addplot+[fill, pattern = grid, restrict x to domain=50:250] table[x=n,y=optimal] {si-table_popularity-P};
        \addlegendentry{(\Profit)}
        \addplot+[fill] table[x=n,y=optimal, restrict x to domain=50:250] {si-table_popularity-F};
        \addlegendentry{(\Utility)}
      \end{axis}
  \end{tikzpicture}
  }
\caption{Number of solutions found by each MIP formulation. We omitted instances of size 300 of the graphs because none of the formulations could solve them within the time limit. For more information, one can refere to Table~\ref{table:optimal} in Appendix~\ref{sec:tables}.}\label{graph:solved}
\end{figure}
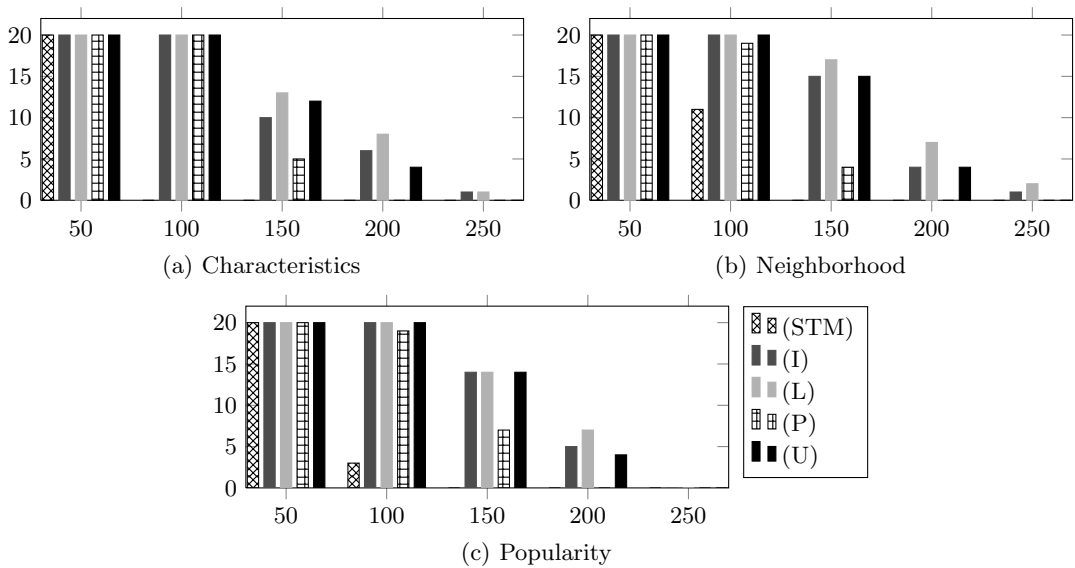

In our experiments, CPLEX had a better performance with formulation (\Loose) than with any other formulation for small instances. First of all, with (\Loose), CPLEX solved more small instances than with the others formulation. For example, with~(\STM), CPLEX could not solve (within one hour) any instance of the Popularity Model
with~$150$ or more items. But, with~(\Loose), CPLEX solved $14$ instances with $150$ items and $7$
instances with $200$ items. 

Also, Table~\ref{table:small_gap} and Fig.~\ref{fig:small_gap} in Appendix~\ref{sec:tables} shows that (\STM) and (\Profit) provided a mean final gap for small instances that where not solved within one hour that were much worse than those provided by (\Utility), (\Improved) and (\Loose) (at least twice the gap for those formulations).

Even thought formulation (\Loose) was better than the others, we noticed that formulations (\Utility) and (\Improved) were not so far behind from formulation (\Loose) as one can observe in Figs.~\ref{graph:solved} (in this section) and~\ref{fig:small_gap} and Table~\ref{table:small_gap} (in Appendix~\ref{sec:tables}).

For large instances, when analyzing the mean final gap for large instances, we conclude that formulation (\Utility) is better than the others. Actually, (\Utility) was able to maintain a really mean small gap even for instances with 3000 bidders and items (the mean final gap was of 2\% for the Characteristics model, 3\% for the Neighborhood model and 17\% for the Popularity model). Figure~\ref{fig:large_gap} and Table~\ref{table:large_gap} in Appendix~\ref{sec:tables} presents this information in more detail.

We believe that this result is impressive because, theoretically (as proven in Theorem~\ref{thm:polyhedra}), (\Utility) is the weakest of our new formulations (we have no result that compares (\Utility) and (\STM)). None the less, (\Utility) is a small formulation when compared with (\STM), (\Improved) and (\Loose) and, because of that, the linear relaxation can be computed faster for this formulation (and also for formulation (\Profit)) as one can see in Fig.~\ref{fig:relax_time} and Table~\ref{table:relax_time} in Appendix~\ref{sec:tables}. Because (\Profit) is a stronger formulation than (\Utility) and they have the same size, we believe that formulation (\Utility) outperformed formulation (\Profit) because of the heuristics used by CPLEX to find good integer solutions. 

\section{Hardness of Approximating Geometric Prices}\label{sec:particular_case}

The result of Briest~\cite{Briest08} shows how hard is to solve efficiently the unit-demand envy-free pricing problem. While pursuing a variant of the problem for which we could design a polynomial-time algorithm or a constant approximation and inspired by~\cite{AFMZ04}, we considered the case where the prices are restricted to the elements of a geometric series. 

Note that this reduces drastically the space of feasible prices and that could potentially lead to solutions far from optimal. This is not the case, as we prove below, the value of an optimal solution of this variant is within a constant factor of the value of an optimal solution of the original problem. Before that, we introduce some useful notation. 

\begin{definition}
  For a fixed valuation $v$ we denote $\max\{v_{ib} : i \in I, b \in B\}$ by $V$. For a fixed positive rational $\epsilon$, we denote by $\mathcal{D}_{1+\epsilon}$ the set $\{d_0, d_1, \ldots \}$ where, for every non-negative integer $k$, $d_k = V/(1+\epsilon)^k$.
\end{definition}

\begin{definition}
  For positive real numbers $n$ and a set $D_{1+\epsilon}$, let $\left\lfloor n \right\rfloor_{1+\epsilon} = d_r$ where~$r$ is the integer such that $d_r \leq n < d_{r-1}$. 
\end{definition}

Remember that in Sect.~\ref{sec:models-and-notation} we defined $x_p$ as an envy-free allocation for $p$ that maximizes the auctioneer's profit.

\begin{definition}
  Let $\sol(p)$ be the auctioneer's profit for a pricing $p$ and~$x_p$. 
\end{definition}

Now, let us first consider that $\epsilon = 1$, that is, the prices are restricted to $D_2$. 
We show that the maximum profit achieved by a pricing restricted by $\mathcal{D}_2$ is at
least a constant fraction of the value of an optimal solution for the unit-demand envy-free pricing
problem.
\begin{theorem} \label{thm:power2} 
  For every pricing $p$, there exists a pricing $\tilde{p}$ consisting of prices only in $\mathcal{D}_2$ such
  that $\sol(\tilde{p}) \geq \sol(p)/4$. 
\end{theorem}
\begin{proof}
  We obtain $\tilde{p}$ by rounding down the components of the vector $2p/3$. That is, $\tilde{p_i}
  = \left\lfloor \frac{2p_i}{3} \right\rfloor_{2}$ for every item $i$. First note that $p_i/3 <
  \tilde{p_i} \leq 2p_i/3$. Let us show that $\sol(\tilde{p}) \geq \sol(p)/4$.
  
  As $\tilde{p_i} < p_i$ for every item $i$, every bidder that receives an item in $x_p$ also
  receives an item in $x_{\tilde{p}}$. Now, suppose that bidder $b$ receives item $i$ in $x_p$ and
  receives item $k$ in $x_{\tilde{p}}$. If $\tilde{p}_k \geq \tilde{p}_i$, then the profit obtained
  from $b$ is at least $\tilde{p}_i > p_i/3$. So, from now on, we assume that $\tilde{p}_k <
  \tilde{p}_i$.
  
  Since $b$ receives $i$ and not $k$ in~$x_p$, we must have that $v_{ib}-p_i \geq v_{kb}-p_k$. Also,
  as~$b$ receives~$k$ and not~$i$ in~$x_{\tilde{p}}$, we must have that $v_{kb}-\tilde{p}_k \geq
  v_{ib}-\tilde{p}_i$, which implies that $p_k - \tilde{p}_k \geq p_i - \tilde{p}_i$. By the
  tie-breaking rule used in the definition of~$x_{p}$, the inequality is in fact strict, 
  as $\tilde{p}_k < \tilde{p}_i$.
  
  Now, suppose that $\tilde{p}_k \leq \tilde{p_i}/4$. As $\tilde{p}_i \leq 2p_i/3$, we have that
  \[p_i - \tilde{p}_i \ \geq \ \frac{3\tilde{p}_i}{2} - \tilde{p}_i 
                        \ = \  \frac{\tilde{p}_i}{2} \ \geq \ 2\tilde{p}_k 
                        \ = \ 3\tilde{p}_k - \tilde{p}_k \ \geq \ p_k - \tilde{p}_k, \]
  contradicting the fact that $p_k - \tilde{p}_k > p_i - \tilde{p}_i$. So $\tilde{p}_k > \tilde{p_i}/4$, 
  but as $\tilde{p}_k$ and $\tilde{p}_i$ are prices in $D_2$, we conclude that $\tilde{p}_k \geq \tilde{p_i}/2$. Hence,
  \[ \tilde{p}_k \ \geq \ \frac{\tilde{p_i}}{2} \ > \ \frac{p_i + \tilde{p}_k - p_k}{2} 
                 \ \geq \ \frac{p_i + \tilde{p}_k - 3\tilde{p}_k}{2}, \]
  and thus $\tilde{p}_k > p_i/4$. Therefore, we conclude that $\sol(\tilde{p}) \geq \sol(p)/4$. \qed 
\end{proof}

One can prove a similar result for pricings that are restricted to $D_{1+\epsilon}$, for an
arbitrary $\eps>0$. The ratio achieved for such pricings with respect to the optimal value
approaches $1$ when $\eps$ goes to $0$. This is stated in the next theorem, whose proof is 
presented in Appendix~\ref{sec:omitted-proofs}.

\newcounter{hardness}
\setcounter{hardness}{\value{theorem}}
\begin{theorem}\label{theo:efp-not-so-far} 
  Let $p$ be a pricing. For each real $\eps$ such that $0 < \eps < 1$, there exists a pricing
  $\tilde{p^\eps}$ consisting of only prices in $D_{1+\epsilon}$ such that $\sol(\tilde{p^\eps})$ goes to
  $\sol(p)$ as $\eps$ goes to $0$. \qed 
\end{theorem}

A consequence of this result is the following. If there exists a constant factor approximation for
the unit-demand envy-free pricing problem with prices restricted to $D_{1+\eps}$ for
some $\epsilon > 0$, then there would exist a constant factor approximation for the unrestricted
problem. This is very unlikely, as shown by Briest~\cite{Briest08}. So, even for this variant, one
cannot hope to find a constant factor approximation without invalidating the hardness assumptions
made by Briest. On the other hand, an empirically fast algorithm for solving or approximating this
variant of the problem would lead to good approximations for the original problem.

\section{Final Remarks}\label{sec:conclusion} 

In this paper, we presented four new MIP formulations for the unit-demand envy-free pricing problem along
with some empirical results comparing such formulations. Our results show that our formulations are better 
than the previous formulation from the literature.

Additionally, we presented a new result regarding the hardness of approximating a variant of the
problem. The fact that, restricting the prices to be chosen from a geometric series does not make the
problem easier to approximate, is really interesting and gives more evidence on how hard it is to
theoretically approach this problem.

We also presented three models for generating instances of unit-demand auctions. To our
knowledge, these are the first tests suites for unit-demand auctions in the literature.
We made the corresponding instances generators available as open-source software.

\bibliographystyle{splncs03} 
\bibliography{EF} 

\begin{thebibliography}{10}
\providecommand{\url}[1]{\texttt{#1}}
\providecommand{\urlprefix}{URL }

\bibitem{AFMZ04}
Aggarwal, G., Feder, T., Motwani, R., Zhu, A.: Algorithms for multi-product
  pricing. In: Proc.\ of the 31th Int.\ Colloquium on Automata, Languages and
  Programming (ICALP). Lecture Notes in Computer Science, vol. 3142, pp. 72--83
  (2004)

\bibitem{BarabasiA99}
Barabási, A.L., Albert, R.: Emergence of scaling in random networks. Science
  286(5439),  509--512 (1999)

\bibitem{Briest08}
Briest, P.: Uniform budgets and the envy-free pricing problem. In: Proc. of the
  35th Int.\ Colloquium on Automata, Languages and Programming (ICALP). Lecture
  Notes in Computer Science, vol. 5125, pp. 808--819 (2008)

\bibitem{CD10}
Chen, N., Deng, X.: Envy-free pricing in multi-item markets. In: Proc. of the
  37th Int.\ Colloquium on Automata, Languages and Programming (ICALP). Lecture
  Notes in Computer Science, vol. 6198/6199, pp. 418--429 (2010)

\bibitem{CramtonSS06}
Cramton, P., Shoham, Y., Steinberg, R. (eds.): Combinatorial Auctions. MIT
  Press (2006)

\bibitem{Gilbert61}
Gilbert, E.N.: Random plane networks. Journal of the Society for Industrial and
  Applied Mathematics  9(4),  533--543 (1961)

\bibitem{GHKKKM05}
Guruswami, V., Hartline, J.D., Karlin, A.R., Kempe, D., Kenyon, C., McSherry,
  F.: On profit-maximizing envy-free pricing. In: Proc. of the 16th Annual
  ACM-SIAM Symposium on Discrete Algorithms (SODA). pp. 1164--1173 (2005)

\bibitem{LeytonBrown00}
Leyton-Brown, K., Pearson, M., Shoham, Y.: Towards a universal test suite for
  combinatorial auction algorithms. In: Proc. of the 2nd ACM Conf.\ on
  Electronic Commerce (EC). pp. 66--76 (2000)

\bibitem{MyklebustST12}
Myklebust, T.G.J., Sharpe, M.A., Tunçel, L.: Efficient heuristic algorithms
  for maximum utility product pricing problems (2012),
  http://www.math.uwaterloo.ca/$\sim$ltuncel/publications/pricedown.pdf

\bibitem{NisanRTV07}
Nisan, N., Roughgarden, T., Tardos, E., Vazirani, V. (eds.): Algorithmic Game
  Theory. Cambridge University Press (2007)

\bibitem{ShiodaTM11}
Shioda, R., Tun\c{c}el, L., Myklebust, T.G.: Maximum utility product pricing
  models and algorithms based on reservation price. Computational Optimization
  and Applications  48(2),  157--198 (Mar 2011)

\end{thebibliography}
\newpage 
\appendix 

\section{Proof of Theorem~\ref{thm:polyhedra}}

We prove Theorem~\ref{thm:polyhedra} using the following lemmas.

\begin{lemma}\label{l:i_better_than_s}
  For every instance $v$ of the unit-demand envy-free problem, $\relaxation_{\Improved}(v)~\leq~\relaxation_{\STM}(v)$.
\end{lemma}
\begin{proof}
  Let $(x,p,\hat{p})$ be a feasible solution of (\Improved). We show that $(x,p,\hat{p})$ is also feasible for (\STM). For that, we only have to prove that, for every bidder $b$ and item $k$, \[\sum_{i \in I \setminus \{k\}} (v_{ib}x_{ib} - \hat{p}_{ib}) \ \geq \ v_{kb}\!\sum_{i \in I \setminus \{k\}}x_{ib} - p_k.\]

  From the feasibility of $(x,p,\hat{p})$, we have that 
  $\sum_{i \in I} (v_{ib}x_{ib} - \hat{p}_{ib}) \geq v_{kb} - p_k$, so we conclude that
  \[\sum_{i \in I\setminus\{k\}} (v_{ib}x_{ib} - \hat{p}_{ib}) \geq v_{kb} - p_k + v_{kb}x_{kb} - \hat{p}_{kb}
  \geq v_{kb}(1-x_{kb}) \ - \ p_k \geq v_{kb}\sum_{i \in I\setminus\{k\}}x_{ib} \ - \ p_k,\]
  from where the result follows.\qed 
\end{proof}

\begin{lemma}\label{l:i_better_than_l}
  For every instance $v$ of the unit-demand envy-free problem, $\relaxation_{\Improved}(v)~\leq~\relaxation_{\Loose}(v)$.
\end{lemma}

\begin{proof}
  The result is clear since (\Loose) has a subset of the inequalities of (\Improved).\qed
\end{proof}

\begin{lemma}\label{l:l_better_than_p}
  For every instance $v$ of the unit-demand envy-free problem, $\relaxation_{\Loose}(v)~\leq~\relaxation_{\Profit}(v)$.
\end{lemma}

\begin{proof}
  Let $(x,p,\hat{p})$ be a feasible solution of the linear relaxation of $(\Loose)$. We prove that $(x,p,z)$, where $z_b = \sum_{i \in I} \hat{p}_{ib}$, is a feasible solution of the linear relaxation of $(\Loose)$ with the same value as $(x,p,\hat{p})$.

  First of all, it is clear that, for every item $i$, we have that $p_i \geq 0$ and also, for every bidder $b$, we have that $\sum_{i \in I} x_{ib} \leq 1$. Finally, for every pair $(i,b)$ in $I \times B$, we have that $x_{ib} \geq 0$.

  Also, for a bidder $b$, using the fact that $z_b = \sum_{i \in I} \hat{p}_{ib}$, it is easy to see that 
  $z_b \geq p_i - R_i(1 - x_{ib})$ (because $\hat{p}_{ib} \geq p_i - R_i(1 - x_{ib})$ for every item $i$), 
  $\sum_{i \in I} v_{ib}x_{ib} - z_b  \geq  0$ (because $v_{ib}x_{ib} - \hat{p}_{ib} \geq 0$ for every item $i$) and that $\sum_{i \in I} v_{ib}x_{ib} - z_b \geq v_{kb} - p_k$ for every item $k$, since 
  $\sum_{i \in I} (v_{ib}x_{ib} - \hat{p}_{ib}) \geq v_{kb} - p_k$.

  Now notice that $(x,p,z)$ has the same value in (\Profit) as $(x,p,\hat{p})$ has in (\Loose). \qed
\end{proof}

\begin{lemma}\label{l:p_better_than_u}
  For every instance $v$ of the unit-demand envy-free problem, $\relaxation_{\Profit}(v)~\leq~\relaxation_{\Utility}(v)$.
\end{lemma}

\begin{proof}
  Let $(x,p,z)$ be a feasible solution of the linear relaxation of $(\Profit)$. We prove that $(x,p,u)$, where 
  $u_b = \sum_{i \in I} v_{ib}x_{ib} - z_b$ is a feasible solution of the linear relaxation of $(\Utility)$ 
  with the same value as $(x,p,z)$.

  Because $(x,p,z)$ is feasible for (\Profit), we have that $\sum_{i \in I} x_{ib} \leq 1$ for every bidder $b$, 
  $p_i \geq 0$ for every item $i$ and, for every pair $(i,b) \in I \times B$, we have that $x_{ib} \geq 0$. 

  Since $z_b \geq 0$ for every bidder $b$, we conclude that $\sum_{i \in I} v_{ib}x_{ib} - u_b \geq 0$, and 
  because $\sum_{i \in I} v_{ib}x_{ib} - z_b \geq  0$, we conclude that $u_b \geq 0$. Also, we have that 
  $u_b = \sum_{i \in I} v_{ib}x_{ib} - z_b \geq v_{kb} - p_k$ for every item~$k$.

  Now, notice that, as $S_b = \max\{v_{ib} : i \in I\}$ and $\sum_{i \in I} x_{ib} \leq 1$, for a bidder $b$ and 
  an item $i$, \[\sum_{i' \in I} v_{i'b}x_{i'b} = v_{ib}x_{ib} + \sum_{i' \in I\setminus\{i\}} v_{i'b}x_{i'b} 
  \leq v_{ib}x_{ib} + \sum_{i' \in I\setminus\{i\}}S_bx_{i'b} \leq v_{ib}x_{ib} + S_b(1-x_{ib}).\]

  From this, we conclude that 
  \begin{align*}
  u_b = \sum_{i' \in I} v_{i'b}x_{i'b} - z_b
  & \leq v_{ib}x_{ib} + S_b(1-x_{ib}) - z_b\\
  & \leq v_{ib}x_{ib} + S_b(1-x_{ib}) - p_i + R_i(1-x_{ib})\\
  & \leq v_{ib}x_{ib} - p_i + (S_b + R_i)(1-x_{ib}),
  \end{align*}
  proving that $(x,p,u)$ is feasible in (\Utility). It is easy to see that $(x,p,u)$ has the same 
  value in (\Utility) as $(x,p,z)$ has in (\Profit). \qed
\end{proof}

\setcounter{aux}{\value{theorem}}
\setcounter{theorem}{\value{polyhedra}}
\begin{theorem}
  For every instance $v$ of the unit-demand envy-free problem, $\relaxation_{\Improved}(v)~\leq~\relaxation_{\STM}(v)$ 
  and $\relaxation_{\Improved}(v) \leq \relaxation_{\Loose}(v) \leq \relaxation_{\Profit}(v) \leq \relaxation_{\Utility}(v)$.
  Also, there is an instance where $\relaxation_{\Improved}(v) < \relaxation_{\STM}(v)$ and 
  $\relaxation_{\Improved}(v) = \relaxation_{\Loose}(v) < \relaxation_{\Profit}(v) < \relaxation_{\Utility}(v)$.
\end{theorem}
\setcounter{theorem}{\value{aux}}

\begin{proof}
  The first part of this proof is given by Lemmas~\ref{l:i_better_than_s},~\ref{l:i_better_than_l},~\ref{l:l_better_than_p}, and~\ref{l:p_better_than_u}. The second part is shown by our empirical results. All the data can be found at \texttt{https://github.com/schouery/unit-demand-market-models}.
\end{proof}

\section{Proof of Theorem~\ref{theo:efp-not-so-far}}\label{sec:omitted-proofs}

Before its proof, let us restate Theorem \ref{theo:efp-not-so-far}. 
\setcounter{aux}{\value{theorem}}
\setcounter{theorem}{\value{hardness}}
\begin{theorem}
  Let $p$ be a pricing. For each real $\eps$ such that $0 < \eps < 1$, there exists a pricing
  $\tilde{p^\eps}$ consisting of only prices in $D_{1+\epsilon}$ such that $\sol(\tilde{p^\eps})$ goes to
  $\sol(p)$ as $\eps$ goes to $0$.
\end{theorem}
\setcounter{theorem}{\value{aux}}

\begin{proof}
  For every bidder $i$, let $\tilde{p^\eps_i} = \left\lfloor \frac{p_i}{r}
  \right\rfloor_{1+\epsilon}$, where $r$ is a constant greater than $1$ that we choose later
  on. Note that $\frac{p_i}{(1+\eps)r} < \tilde{p^\eps_i} \leq \frac{p_i}r$.
  
  As in the proof of Theorem~\ref{thm:power2}, every bidder that receives an item in $x_p$ also
  receives an item in~$x_{\tilde{p^\eps}}$, and if bidder $b$ exchanges item~$i$ for item $k$ when
  the pricing changes to $\tilde{p^\eps}$, then $p_k - \tilde{p^\eps_k} > p_i - \tilde{p^\eps_i}$
  unless $\tilde{p^\eps_k} \geq \tilde{p^\eps_i}$.

  If $\tilde{p^\eps_k} \geq \tilde{p^\eps_i}$, then $\tilde{p^\eps_k} > \frac{p_i}{(1+\eps)r}$.
  Otherwise \[ \tilde{p^\eps_k} \geq \frac{p_k}{(1+\epsilon)r} > \frac{\tilde{p^\eps_k} + p_i - 
  \tilde{p^\eps_i}}{(1+\epsilon)r} \geq \frac{\tilde{p^\eps_k} + p_i - \frac{p_i}{r}}{(1+\epsilon)r}, \]
  which implies that
  \[ \tilde{p^\eps_k} \geq \frac{r-1}{(1+\epsilon)r^2 - r} p_i. \]
  It can be shown that the maximum on the right side is achieved at $r = 1 +
  \frac{\sqrt{\epsilon}}{\sqrt{1+\epsilon}}$. Using this value of $r$, we have that
  \[\tilde{p^\eps_k} \geq \frac{1}{2\sqrt{\eps(1+\eps)} + 2\eps + 1}p_i. \]
  As $\eps$ goes to $0$, the chosen value of $r$ goes to $1$, and $\sol(\tilde{p^\eps})$ goes to
  $\sol(p)$.~\qed
\end{proof}

\section{Tables and Graphics}\label{sec:tables}

In this section we present our empirical results.

\begin{table}[htb!]
\centering
\subtable[Characteristics]{
\begin{tabular}{|c|c|c|c|c|c|}
\hline
$|I|$ & (\STM)    & (\Improved)   & (\Loose)    & (\Profit)   & (\Utility)    \\
\hline
50    & 20    & 20    & 20    & 20    & 20    \\
\hline
100   & 0     & 20    & 20    & 20    & 20    \\
\hline
150   & 0     & 10    & 13    & 5     & 12    \\
\hline
200   & 0     & 6     & 8     & 0     & 4     \\
\hline
250   & 0     & 1     & 1     & 0     & 0     \\
\hline
300   & 0     & 0     & 0     & 0     & 0     \\
\hline
\end{tabular}
}
\subtable[Neighborhood]{
\begin{tabular}{|c|c|c|c|c|c|}
\hline
$|I|$ & (\STM)    & (\Improved)   & (\Loose)    & (\Profit)   & (\Utility)    \\
\hline
50    & 20    & 20    & 20    & 20    & 20    \\
\hline
100   & 11    & 20    & 20    & 19    & 20    \\
\hline
150   & 0     & 15    & 17    & 4     & 15    \\
\hline
200   & 0     & 4     & 7     & 0     & 4     \\
\hline
250   & 0     & 1     & 2     & 0     & 0     \\
\hline
300   & 0     & 0     & 0     & 0     & 0     \\
\hline
\end{tabular}
}
\subtable[Popularity]{
\begin{tabular}{|c|c|c|c|c|c|}
\hline
$|I|$ & (\STM)    & (\Improved)   & (\Loose)    & (\Profit)   & (\Utility)    \\
\hline
50    & 20    & 20    & 20    & 20    & 20    \\
\hline
100   & 3     & 20    & 20    & 19    & 20    \\
\hline
150   & 0     & 14    & 14    & 7     & 14    \\
\hline
200   & 0     & 5     & 7     & 0     & 4     \\
\hline
250   & 0     & 0     & 0     & 0     & 0     \\
\hline
300   & 0     & 0     & 0     & 0     & 0     \\
\hline
\end{tabular}
}
\caption{Number of optimal solutions found for small instances in each of the models.}\label{table:optimal}
\end{table}

\begin{figure}[ht!]
\centering
\subfigure[Characteristics]{
  \begin{tikzpicture}
    \begin{axis}[ylabel = Mean Final Gap, 
                 xlabel = $|I|$,
                 width = 7.5cm,
                 height = 6cm,
                 xtick = {25, 50, ..., 300},
                 legend style={legend pos=outer north east,}
               ]
      \addplot[thick, gray] table[x=n,y=gap] {si-n-table_characteristics-S};
      \addplot[gray, mark=*] table[x=n,y=gap] {si-n-table_characteristics-I};
      \addplot[thick, gray, mark=+] table[x=n,y=gap] {si-n-table_characteristics-L};
      \addplot[thick, black, mark=+] table[x=n,y=gap] {si-n-table_characteristics-P};
      \addplot[black] table[x=n,y=gap] {si-n-table_characteristics-F};
    \end{axis}
  \end{tikzpicture}
}
\subfigure[Neighborhood]{
  \begin{tikzpicture}
    \begin{axis}[ylabel = Mean Final Gap, 
                 xlabel = $|I|$,
                 width = 7.5cm,
                 height = 6cm,
                 xtick = {25, 50, ..., 300},
                 legend style={legend pos=outer north east,}
               ]
      \addplot[thick, gray] table[x=n,y=gap] {si-n-table_neighborhood-S};
      \addplot[gray, mark=*] table[x=n,y=gap] {si-n-table_neighborhood-I};
      \addplot[thick, gray, mark=+] table[x=n,y=gap] {si-n-table_neighborhood-L};
      \addplot[thick, black, mark=+] table[x=n,y=gap] {si-n-table_neighborhood-P};
      \addplot[black] table[x=n,y=gap] {si-n-table_neighborhood-F};
    \end{axis}
  \end{tikzpicture}
}
\subfigure[Popularity]{
  \begin{tikzpicture}
    \begin{axis}[ylabel = Mean Final Gap, 
                 xlabel = $|I|$,
                 width = 7.5cm,
                 height = 6cm,
                 xtick = {25, 50, ..., 300},
                 legend style={legend pos=outer north east,},
                 every axis legend/.append style={nodes={right}}
               ]
      \addplot[thick, gray] table[x=n,y=gap] {si-n-table_popularity-S};
      \addlegendentry{(\STM)}
      \addplot[gray, mark=*] table[x=n,y=gap] {si-n-table_popularity-I};
      \addlegendentry{(\Improved)}
      \addplot[thick, gray, mark=+] table[x=n,y=gap] {si-n-table_popularity-L};
      \addlegendentry{(\Loose)}
      \addplot[thick, black, mark=+] table[x=n,y=gap] {si-n-table_popularity-P};
      \addlegendentry{(\Profit)}
      \addplot[black] table[x=n,y=gap] {si-n-table_popularity-F};
      \addlegendentry{(\Utility)}
    \end{axis}
  \end{tikzpicture}
}
\caption{Mean final gap for small instances (for the three models) that were not solved within the time limit.}
\label{fig:small_gap}
\end{figure}
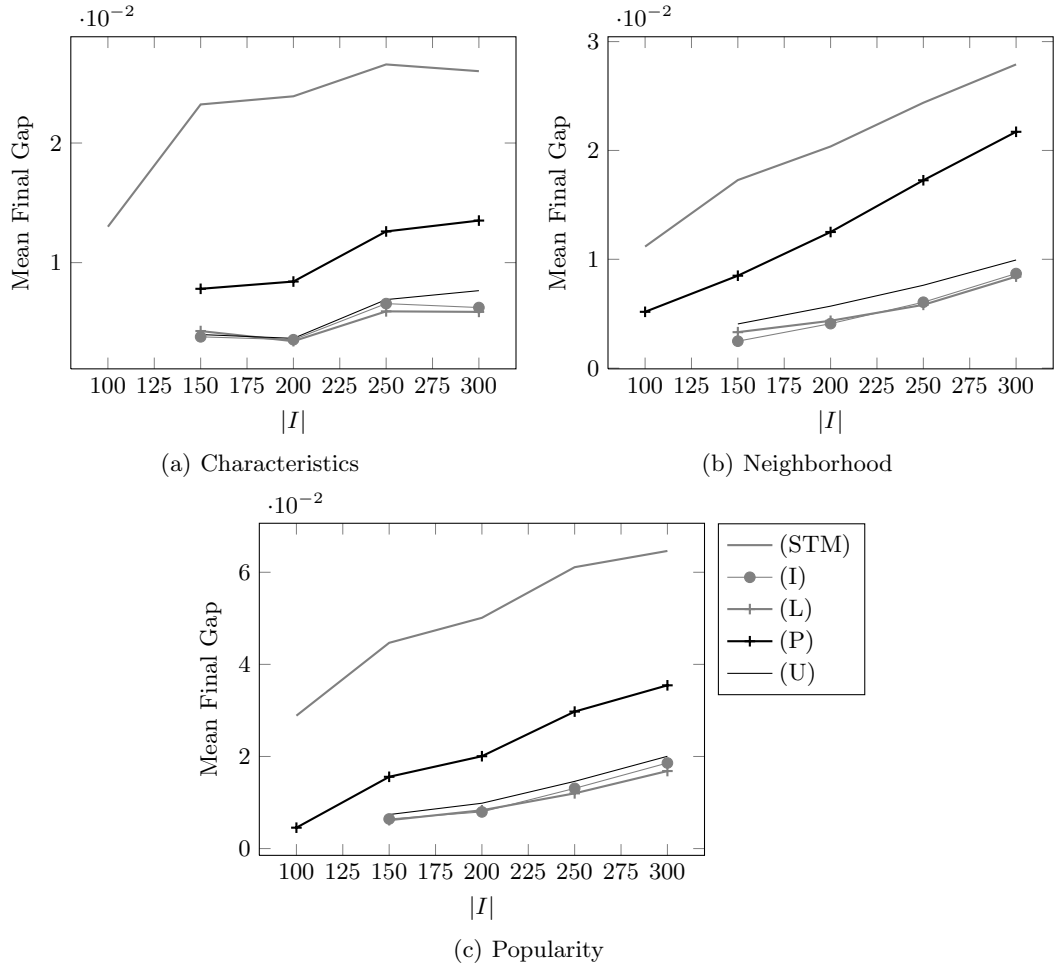

\begin{table}[htb!]
\centering
\subtable[Characteristics]{
\begin{tabular}{|c|c|c|c|c|c|}
\hline
$|I|$ & (\STM)    & (\Improved)   & (\Loose)    & (\Profit)   & (\Utility)    \\
\hline
50    & -     & -     & -     & -     & -     \\
\hline
100   & 0.013 (0.00648)   & -     & -     & -     & -     \\
\hline
150   & 0.0232 (0.00678)    & 0.0038 (0.00257)    & 0.00428 (0.00202)   & 0.00781 (0.00348)   & 0.00399 (0.00248)   \\
\hline
200   & 0.0239 (0.00417)    & 0.00356 (0.00145)   & 0.00343 (0.000715)    & 0.00842 (0.00287)   & 0.00367 (0.00152)   \\
\hline
250   & 0.0266 (0.00433)    & 0.00657 (0.00257)   & 0.00593 (0.00259)   & 0.0126 (0.00352)    & 0.00691 (0.00271)   \\
\hline
300   & 0.026 (0.00439)   & 0.00624 (0.00243)   & 0.00588 (0.00245)   & 0.0135 (0.00344)    & 0.00766 (0.00257)   \\
\hline
\end{tabular}
}
\subtable[Neighborhood]{
\begin{tabular}{|c|c|c|c|c|c|}
\hline
$|I|$ & (\STM)    & (\Improved)   & (\Loose)    & (\Profit)   & (\Utility)    \\
\hline
50    & -     & -     & -     & -     & -     \\
\hline
100   & 0.0112 (0.0062)   & -     & -     & 0.00517 (0)   & -     \\
\hline
150   & 0.0173 (0.00646)    & 0.00248 (0.00247)   & 0.00331 (0.00021)   & 0.00849 (0.00416)   & 0.00406 (0.00168)   \\
\hline
200   & 0.0204 (0.00498)    & 0.00409 (0.0026)    & 0.00435 (0.0024)    & 0.0125 (0.00471)    & 0.00569 (0.00191)   \\
\hline
250   & 0.0244 (0.00636)    & 0.00606 (0.00351)   & 0.00581 (0.00357)   & 0.0173 (0.00422)    & 0.00761 (0.00375)   \\
\hline
300   & 0.0279 (0.00442)    & 0.00869 (0.00245)   & 0.00841 (0.00257)   & 0.0217 (0.00415)    & 0.00994 (0.00289)   \\
\hline
\end{tabular}
}
\subtable[Popularity]{
\begin{tabular}{|c|c|c|c|c|c|}
\hline
$|I|$ & (\STM)    & (\Improved)   & (\Loose)    & (\Profit)   & (\Utility)    \\
\hline
50    & -     & -     & -     & -     & -     \\
\hline
100   & 0.0289 (0.0138)   & -     & -     & 0.00454 (0)   & -     \\
\hline
150   & 0.0447 (0.0151)   & 0.00643 (0.0024)    & 0.00616 (0.00355)   & 0.0156 (0.00801)    & 0.00739 (0.00362)   \\
\hline
200   & 0.0501 (0.013)    & 0.00796 (0.00424)   & 0.00832 (0.00454)   & 0.0201 (0.00858)    & 0.00985 (0.00475)   \\
\hline
250   & 0.0611 (0.0093)   & 0.0131 (0.00547)    & 0.012 (0.0054)    & 0.0297 (0.00687)    & 0.0146 (0.00485)    \\
\hline
300   & 0.0646 (0.0132)   & 0.0186 (0.00749)    & 0.0168 (0.00708)    & 0.0355 (0.00864)    & 0.02 (0.00683)    \\
\hline
\end{tabular}
}
\caption{Mean final gap for small instances (for the three models) that were not solved within the time limit. The value in parenthesis denotes the standard deviation.}\label{table:small_gap}
\end{table}

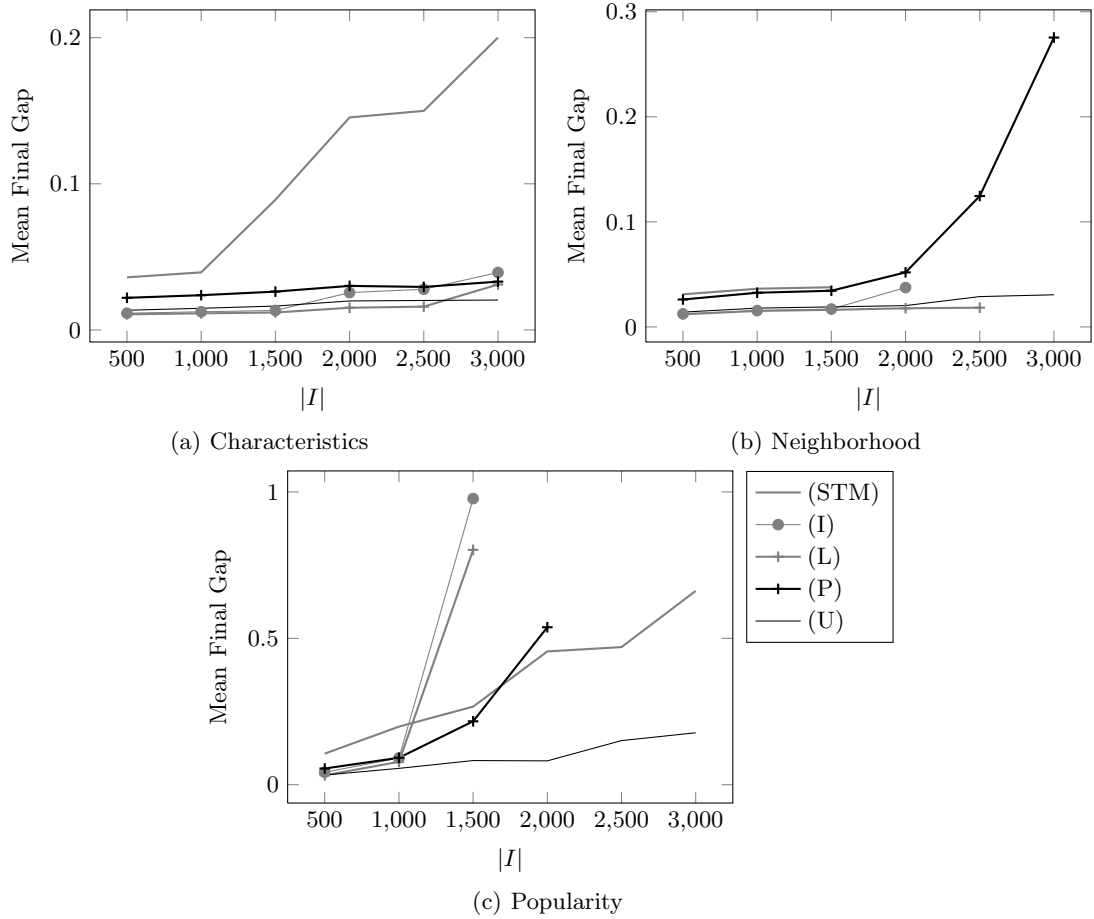
\begin{figure}[ht!]
\centering
\subfigure[Characteristics]{
  \begin{tikzpicture}
    \def\ylimit{1}
    \begin{axis}[ylabel = Mean Final Gap, 
                 xlabel = $|I|$,
                 width = 7.5cm,
                 height = 6cm,
                 xtick = {500, 1000, ..., 3000},
                 legend style={legend pos=outer north east,}
               ]
      \addplot[thick, gray, restrict y to domain=0:\ylimit] table[x=n,y=gap] {li-table_characteristics-S};
      \addplot[gray, mark=*, restrict y to domain=0:\ylimit] table[x=n,y=gap] {li-table_characteristics-I};
      \addplot[thick, gray, mark=+, restrict y to domain=0:\ylimit] table[x=n,y=gap] {li-table_characteristics-L};
      \addplot[thick, black, mark=+, restrict y to domain=0:\ylimit] table[x=n,y=gap] {li-table_characteristics-P};
      \addplot[black] table[x=n,y=gap, restrict y to domain=0:\ylimit] {li-table_characteristics-F};
    \end{axis}
  \end{tikzpicture}
}
\subfigure[Neighborhood]{
  \begin{tikzpicture}
    \def\ylimit{1}
    \begin{axis}[ylabel = Mean Final Gap, 
                 xlabel = $|I|$,
                 width = 7.5cm,
                 height = 6cm,
                 xtick = {500, 1000, ..., 3000},
                 legend style={legend pos=outer north east,}
               ]
      \addplot[thick, gray, restrict y to domain=0:\ylimit] table[x=n,y=gap] {li-table_neighborhood-S};
      \addplot[gray, mark=*, restrict y to domain=0:\ylimit] table[x=n,y=gap] {li-table_neighborhood-I};
      \addplot[thick, gray, mark=+, restrict y to domain=0:\ylimit] table[x=n,y=gap] {li-table_neighborhood-L};
      \addplot[thick, black, mark=+, restrict y to domain=0:\ylimit] table[x=n,y=gap] {li-table_neighborhood-P};
      \addplot[black] table[x=n,y=gap, restrict y to domain=0:\ylimit] {li-table_neighborhood-F};
    \end{axis}
  \end{tikzpicture}
}
\subfigure[Popularity]{
  \begin{tikzpicture}
    \def\ylimit{1}
    \begin{axis}[ylabel = Mean Final Gap, 
                 xlabel = $|I|$,
                 width = 7.5cm,
                 height = 6cm,
                 xtick = {500, 1000, ..., 3000},
                 legend style={legend pos=outer north east,},
                 every axis legend/.append style={nodes={right}}
               ]
      \addplot[thick, gray] table[x=n,y=gap, restrict y to domain=0:\ylimit] {li-table_popularity-S};
      \addlegendentry{(\STM)}
      \addplot[gray, mark=*] table[x=n,y=gap, restrict y to domain=0:\ylimit] {li-table_popularity-I};
      \addlegendentry{(\Improved)}
      \addplot[thick, gray, mark=+] table[x=n,y=gap, restrict y to domain=0:\ylimit] {li-table_popularity-L};
     \addlegendentry{(\Loose)}
      \addplot[thick, black, mark=+] table[x=n,y=gap, restrict y to domain=0:\ylimit] {li-table_popularity-P};
      \addlegendentry{(\Profit)}
      \addplot[black] table[x=n,y=gap, restrict y to domain=0:\ylimit] {li-table_popularity-F};
      \addlegendentry{(\Utility)}
    \end{axis}
  \end{tikzpicture}
}
\caption{Mean final gap for large instances (for the three models) that were not solved within the time limit. We limited the graphic to plot only the final gaps less or equal to 1, that is, if the mean final gap of a formulation is missing in the graphic then it was bigger than 1.}\label{fig:large_gap}
\end{figure}

\begin{table}[htb!]
\centering
\subtable[Characteristics]{
\begin{tabular}{|c|c|c|c|c|c|}
\hline
$|I|$ & (\STM)    & (\Improved)   & (\Loose)    & (\Profit)   & (\Utility)    \\
\hline
500   & 0.036 (0.0044)    & 0.0114 (0.00194)    & 0.0107 (0.00185)    & 0.022 (0.00292)   & 0.0135 (0.00227)    \\
\hline
1000    & 0.0394 (0.00478)    & 0.0124 (0.00184)    & 0.0114 (0.00127)    & 0.0238 (0.00162)    & 0.0148 (0.00121)    \\
\hline
1500    & 0.0891 (0.0434)   & 0.0133 (0.00172)    & 0.0119 (0.00138)    & 0.0262 (0.00207)    & 0.0163 (0.00303)    \\
\hline
2000    & 0.145 (0.0208)    & 0.0256 (0.0105)   & 0.0152 (0.00267)    & 0.0301 (0.00196)    & 0.0198 (0.00238)    \\
\hline
2500    & 0.15 (0.0212)   & 0.0278 (0.00734)    & 0.0159 (0.00458)    & 0.0294 (0.00199)    & 0.0202 (0.00221)    \\
\hline
3000    & 0.2 (0.0543)    & 0.0393 (0.0102)   & 0.0311 (0.00626)    & 0.0331 (0.00408)    & 0.0205 (0.00269)    \\
\hline
\end{tabular}
}
\subtable[Neighborhood]{
\begin{tabular}{|c|c|c|c|c|c|}
\hline
$|I|$ & (\STM)    & (\Improved)   & (\Loose)    & (\Profit)   & (\Utility)    \\
\hline
500   & 0.031 (0.00614)   & 0.0124 (0.00351)    & 0.012 (0.00331)   & 0.0261 (0.00567)    & 0.0141 (0.00362)    \\
\hline
1000    & 0.0363 (0.00388)    & 0.0156 (0.00195)    & 0.0152 (0.00189)    & 0.0325 (0.00305)    & 0.0179 (0.00212)    \\
\hline
1500    & 0.0377 (0.00259)    & 0.017 (0.00253)   & 0.0163 (0.00177)    & 0.0344 (0.00253)    & 0.0191 (0.00185)    \\
\hline
2000    & $\infty$   & 0.0375 (0.0842)   & 0.0177 (0.00267)    & 0.0519 (0.0383)   & 0.0203 (0.00185)    \\
\hline
2500    & $\infty$   & 31.8 (138)    & 0.0183 (0.00229)    & 0.125 (0.195)   & 0.0289 (0.0306)   \\
\hline
3000    & $\infty$   & 443 (671)   & 26.3 (114)    & 0.275 (0.218)   & 0.0306 (0.0114)   \\
\hline
\end{tabular}
}
\subtable[Popularity]{
\begin{tabular}{|c|c|c|c|c|c|}
\hline
$|I|$ & (\STM)    & (\Improved)   & (\Loose)    & (\Profit)   & (\Utility)    \\
\hline
500   & 0.106 (0.0177)    & 0.0426 (0.0101)   & 0.032 (0.00923)   & 0.0554 (0.0108)   & 0.033 (0.0091)    \\
\hline
1000    & 0.198 (0.0648)    & 0.091 (0.0315)    & 0.0782 (0.0289)   & 0.092 (0.059)   & 0.0559 (0.0156)   \\
\hline
1500    & 0.267 (0.0953)    & 0.977 (2.46)    & 0.802 (2.49)    & 0.217 (0.128)   & 0.0825 (0.0461)   \\
\hline
2000    & 0.455 (0.16)    & 2.18 (3.97)   & 1.32 (3.05)   & 0.538 (0.783)   & 0.0815 (0.018)    \\
\hline
2500    & 0.47 (0.202)    & 2.86 (4.48)   & 3.28 (4.87)   & $\infty$   & 0.151 (0.0706)    \\
\hline
3000    & 0.662 (0.374)   & 3.59 (4.77)   & 2.98 (4.38)   & $\infty$   & 0.177 (0.0637)    \\
\hline
\end{tabular}
}
\caption{Mean final gap for large instances (for the three models) that were not solved within the time limit. The value in parenthesis denotes the standard deviation.}\label{table:large_gap}
\end{table}

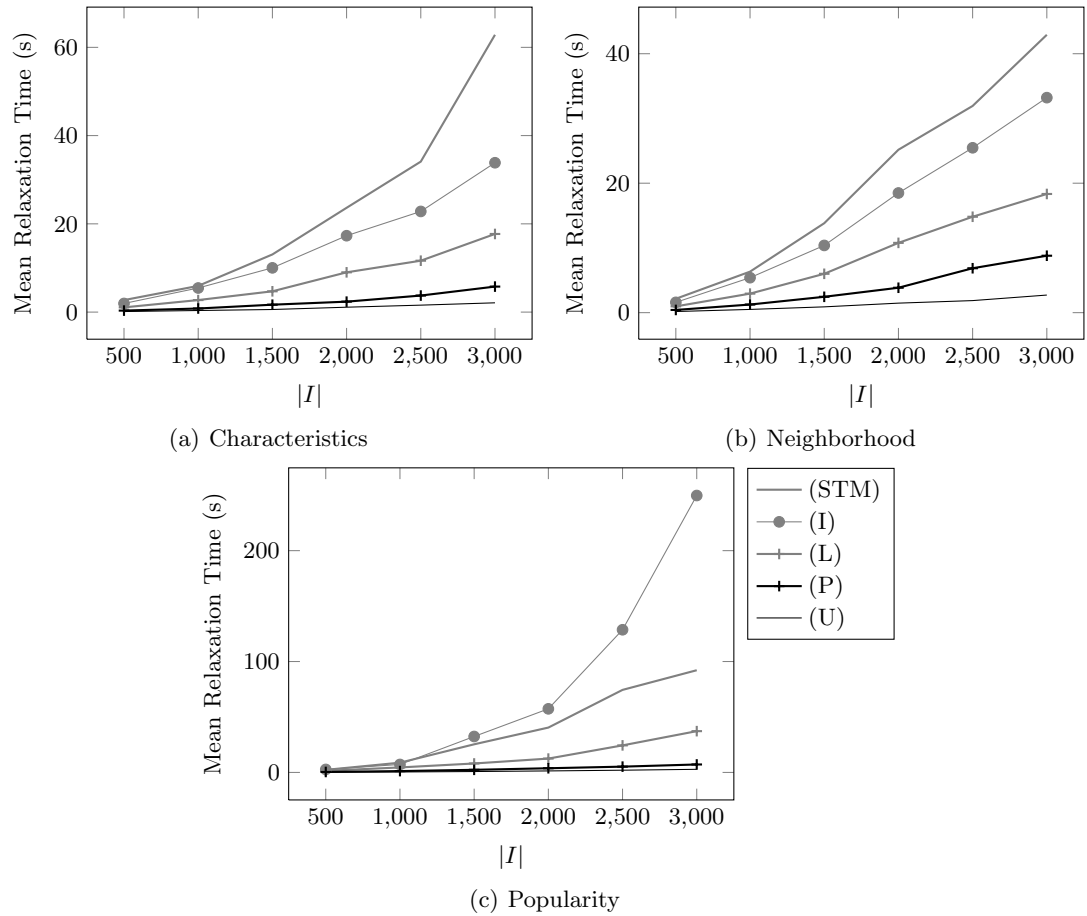
\begin{figure}[ht!]
\centering
\subfigure[Characteristics]{
  \begin{tikzpicture}
    \begin{axis}[ylabel = Mean Relaxation Time (s), 
                 xlabel = $|I|$,
                 width = 7.5cm,
                 height = 6cm,
                 xtick = {500, 1000, ..., 3000},
                 legend style={legend pos=outer north east,},
                 every axis legend/.append style={nodes={right}}
               ]
      \addplot[thick, gray] table[x=n,y=relax_time] {li-table_characteristics-S};
      \addplot[gray, mark=*] table[x=n,y=relax_time] {li-table_characteristics-I};
      \addplot[thick, gray, mark=+] table[x=n,y=relax_time] {li-table_characteristics-L};
      \addplot[thick, black, mark=+] table[x=n,y=relax_time] {li-table_characteristics-P};
      \addplot[black] table[x=n,y=relax_time] {li-table_characteristics-F};
    \end{axis}
  \end{tikzpicture}
}
\subfigure[Neighborhood]{
  \begin{tikzpicture}
    \begin{axis}[ylabel = Mean Relaxation Time (s), 
                 xlabel = $|I|$,
                 width = 7.5cm,
                 height = 6cm,
                 xtick = {500, 1000, ..., 3000},
                 legend style={legend pos=outer north east,}
               ]
      \addplot[thick, gray] table[x=n,y=relax_time] {li-table_neighborhood-S};
      \addplot[gray, mark=*] table[x=n,y=relax_time] {li-table_neighborhood-I};
      \addplot[thick, gray, mark=+] table[x=n,y=relax_time] {li-table_neighborhood-L};
      \addplot[thick, black, mark=+] table[x=n,y=relax_time] {li-table_neighborhood-P};
      \addplot[black] table[x=n,y=relax_time] {li-table_neighborhood-F};
    \end{axis}
  \end{tikzpicture}
}
\subfigure[Popularity]{
  \begin{tikzpicture}
    \begin{axis}[ylabel = Mean Relaxation Time (s), 
                 xlabel = $|I|$,
                 width = 7.5cm,
                 height = 6cm,
                 xtick = {500, 1000, ..., 3000},
                 legend style={legend pos=outer north east,},
                 every axis legend/.append style={nodes={right}}
               ]
      \addplot[thick, gray] table[x=n,y=relax_time] {li-table_popularity-S};
      \addlegendentry{(\STM)}
      \addplot[gray, mark=*] table[x=n,y=relax_time] {li-table_popularity-I};
      \addlegendentry{(\Improved)}
      \addplot[thick, gray, mark=+] table[x=n,y=relax_time] {li-table_popularity-L};
      \addlegendentry{(\Loose)}
      \addplot[thick, black, mark=+] table[x=n,y=relax_time] {li-table_popularity-P};
      \addlegendentry{(\Profit)}
      \addplot[black] table[x=n,y=relax_time] {li-table_popularity-F};
      \addlegendentry{(\Utility)}
    \end{axis}
  \end{tikzpicture}
}
\caption{Mean time (in seconds) to solve the linear relaxation of the first node for large instances (for the three models).}
\label{fig:relax_time}
\end{figure}

\begin{table}[htb!]
\centering
\subtable[Characteristics]{
\begin{tabular}{|c|c|c|c|c|c|}
\hline
$|I|$ & (\STM)    & (\Improved)   & (\Loose)    & (\Profit)   & (\Utility)    \\
\hline
500   & 2.72 (0.749)    & 1.96 (0.485)    & 1.06 (0.264)    & 0.339 (0.071)   & 0.128 (0.024)   \\
\hline
1000    & 5.94 (1.2)    & 5.46 (1.23)   & 2.73 (0.708)    & 0.814 (0.227)   & 0.361 (0.072)   \\
\hline
1500    & 13.1 (2.7)    & 10 (2.1)    & 4.71 (1.11)   & 1.68 (0.422)    & 0.588 (0.148)   \\
\hline
2000    & 23.6 (4.03)   & 17.3 (3.65)   & 9.01 (1.9)    & 2.37 (0.709)    & 1.11 (0.22)   \\
\hline
2500    & 34.1 (10.7)   & 22.8 (3.04)   & 11.6 (2.29)   & 3.75 (0.865)    & 1.58 (0.351)    \\
\hline
3000    & 62.8 (22.5)   & 33.8 (8.34)   & 17.7 (3.09)   & 5.78 (1.54)   & 2.1 (0.512)   \\
\hline
\end{tabular}
}
\subtable[Neighborhood]{
\begin{tabular}{|c|c|c|c|c|c|}
\hline
$|I|$ & (\STM)    & (\Improved)   & (\Loose)    & (\Profit)   & (\Utility)    \\
\hline
500   & 2.15 (0.511)    & 1.59 (0.345)    & 0.969 (0.245)   & 0.42 (0.0757)   & 0.175 (0.0389)    \\
\hline
1000    & 6.34 (1.72)   & 5.4 (1.38)    & 2.94 (0.737)    & 1.25 (0.241)    & 0.501 (0.0883)    \\
\hline
1500    & 13.8 (3.16)   & 10.4 (2.32)   & 6 (1.04)    & 2.45 (0.697)    & 0.897 (0.18)    \\
\hline
2000    & 25.2 (5)    & 18.5 (3.93)   & 10.8 (1.88)   & 3.85 (0.907)    & 1.47 (0.304)    \\
\hline
2500    & 31.9 (6.79)   & 25.5 (4.42)   & 14.8 (2.73)   & 6.86 (1.03)   & 1.86 (0.531)    \\
\hline
3000    & 42.9 (8.13)   & 33.2 (6.77)   & 18.3 (3.38)   & 8.8 (1.9)   & 2.72 (0.549)    \\
\hline
\end{tabular}
}
\subtable[Popularity]{
\begin{tabular}{|c|c|c|c|c|c|}
\hline
$|I|$     & (\STM)    & (\Improved)   & (\Loose)    & (\Profit)   & (\Utility)    \\
\hline
500     & 2.3 (0.589)   & 2.69 (0.662)    & 1.43 (0.231)    & 0.412 (0.0788)    & 0.143 (0.0341)    \\
\hline
1000    & 8.78 (3.1)    & 7.18 (2.18)   & 4.49 (0.667)    & 1.2 (0.242)   & 0.409 (0.0731)    \\
\hline
1500    & 25.4 (15.3)   & 32.4 (43.2)   & 8.06 (2.43)   & 2.34 (0.615)    & 0.799 (0.17)    \\
\hline
2000    & 40.4 (24.6)   & 57.3 (126)    & 12.5 (3.25)   & 3.76 (1.13)   & 1.35 (0.332)    \\
\hline
2500    & 74.4 (51.2)   & 129 (262)   & 24.4 (10.8)   & 5.21 (1.22)   & 1.99 (0.48)   \\
\hline
3000    & 92.2 (59)   & 250 (389)   & 37.2 (19.4)   & 7.17 (2.22)   & 2.8 (0.756)   \\
\hline
\end{tabular}
}
\caption{Mean time (in seconds) to solve the linear relaxation of the first node for large instances (for the three models).
The value in parenthesis denotes the standard deviation.}\label{table:relax_time}
\end{table}
\end{document}